\documentclass[10pt]{scrartcl}

\usepackage[english]{babel}

\usepackage[top=2.5cm, bottom=2.5cm, left=2.5cm, right=2.5cm]{geometry} 

\usepackage{bbm}
\usepackage{authblk}
\usepackage{amsthm}
\usepackage{amsmath}
\usepackage{amssymb}
\usepackage{mathtools}
\usepackage{graphicx}
\usepackage[colorlinks=true, allcolors=blue]{hyperref}
\usepackage{cleveref}
\usepackage{subfigure}
\usepackage{algorithm2e}
\usepackage{enumerate}

\crefname{ineq}{inequality}{inequalities}
\creflabelformat{ineq}{#2{\upshape(#1)}#3}

\usepackage{csquotes}

\usepackage{tikz}
\usepackage{tkz-graph}
\usetikzlibrary{calc}
\usetikzlibrary{fit, shapes.geometric}
\tikzstyle{vertex}=[circle, fill, inner sep=0pt, minimum size=3pt]
\newcommand{\vertex}{\node[vertex]}

\tikzset{snake it/.style={decorate, decoration=snake}}
\tikzset{
dot/.style = {circle, fill, minimum size=#1,
              inner sep=0pt, outer sep=0pt},
dot/.default = 3pt 
}

\usepackage{thmtools}
\usepackage{thm-restate}

\theoremstyle{definition}
\newtheorem{definition}{Definition}

\theoremstyle{remark}

\theoremstyle{plain}
\newtheorem*{fact}{Fact}

\newtheorem*{conjecture}{Conjecture}
\declaretheorem[name=Theorem,numberwithin=section]{thm}

\newtheorem{lemma}[thm]{Lemma}
\newtheorem*{lemma*}{Lemma}

\newcommand*{\myproofname}{Proof}
\newcommand*{\proofsketchname}{Proof sketch}

\let\given\givenbase

\newcommand{\diag}{\operatorname{diag}}

\def\prob{\ensuremath\mathbb{P}}
\def\expect{\ensuremath\mathbb{E}}
\newcommand*\dd{\mathop{}\!\mathrm{d}}

\def\S{\ensuremath\mathcal{S}}
\def\G{\ensuremath\mathcal{G}}
\def\real{\ensuremath\mathbb{R}}
\def\Pol{{\sffamily P}}
\def\sharpP{{\sffamily \#P}}
\def\NP{{\sffamily NP}}
\def\FP{{\sffamily FP}}

\def\twoopt{{\textsc{\#2Opt}}}

\def\hampath{{\textsc{\#HamPath}}}

\let\originalleft\left
\let\originalright\right
\renewcommand{\left}{\mathopen{}\mathclose\bgroup\originalleft}
\renewcommand{\right}{\aftergroup\egroup\originalright}

\title{Counting Locally Optimal Tours in the TSP}

\author[1]{Bodo Manthey}
\author[1]{Jesse van Rhijn}
\affil[1]{Department of Applied Mathematics, University of Twente}

\begin{document}
\maketitle

\begin{abstract}
    We show that the problem of counting the number of 2-optimal tours in instances of
    the Travelling Salesperson Problem (TSP)
    on complete graphs is \sharpP{}-complete. In addition, we show that
    the expected number of 2-optimal tours in random instances
    of the TSP on complete graphs is~$O(1.2098^n \sqrt{n!})$.
    Based on numerical experiments, we conjecture that the true bound is
    at most~$O(\sqrt{n!})$, which is approximately the square root of the total
    number of tours.
\end{abstract}

\section{Introduction}

The Travelling Salesperson Problem is among the best-studied problems in computer science.
It can be stated compactly: given a weighted graph~$G = (V, E)$ with edge weights~$w: E \to \real$, find the Hamiltonian cycle (tour) on~$G$ with the smallest total weight.
The TSP is a classic example of a hard optimization problem, being
even among Karp's original 21 {\sffamily{NP}}-hard problems~\cite{karpReducibilityCombinatorialProblems1972}.

Owing to this hardness, practitioners often turn to approximate methods. One
extremely successful method is local search~\cite{linEffectiveHeuristicAlgorithm1973}. 
This is a general optimization framework
where one modifies an existing (sub-optimal) solution into a better solution.

The simplest local search heuristic for the TSP is 2-opt~\cite{aartsLocalSearchCombinatorial2003}. 
This heuristic takes as its input a tour~$T$,
and finds two sets of two edges each,~$\{e_1, e_2\} \subseteq T$ and~$\{f_1, f_2\}
\nsubseteq T$, such that exchanging~$\{e_1, e_2\}$ for~$\{f_1, f_2\}$
yields again a tour~$T'$, and the total weight of~$T'$ is strictly
less than the total weight of~$T$. This procedure is repeated with the new tour,
and stops once no such edges exist. The resulting tour is
said to be locally optimal with respect to the 2-opt neighborhood.

A convenient way to view 2-opt (and other local search heuristics) is via the transition graph~$\mathcal{T}$.
This directed graph contains a node for every tour of~$G$. An arc~$(T_1, T_2)$
exists in~$\mathcal{T}$ if and only if~$T_2$ can be obtained from~$T_1$
by a 2-opt step and~$T_2$ has strictly lower cost than~$T_1$. 
The sinks of~$\mathcal{T}$ are exactly the locally optimal tours of~$G$.
A run of 2-opt
can then be characterized by a directed path through~$\mathcal{T}$ ending in a sink.

Much research has previously focused on understanding the running time of
2-opt~\cite{chandraNewResultsOld1999,englertWorstCaseProbabilistic2014,englertSmoothedAnalysis2Opt2016,mantheyImprovedSmoothedAnalysis2023b,mantheySmoothedAnalysis2Opt2013}
and its approximation ratio~\cite{brodowskyApproximationRatio2Opt2021,engelsAveragecaseApproximationRatio2009,englertWorstCaseProbabilistic2014,hougardyApproximationRatio2Opt2020,kunnemannSmoothedAnalysis2Opt2023}. On the other
hand, little is known about the structure of~$\mathcal{T}$.
In this paper, we are
concerned with counting the number of sinks of~$\mathcal{T}$, which is equivalent
to counting the number of 2-optimal tours in the instance represented by~$\mathcal{T}$.

There are practical reasons to study the transition graphs of local search heuristics.
First, observe that the transition graph of 2-opt for the random TSP instances we
consider is a type of random directed acyclic graph. 
A run of
2-opt can be viewed as a path through this random DAG, so that the
the length of the longest path in~$\mathcal{T}$ is an upper bound for the number of iterations
2-opt can perform. This upper bound is however rather crude: If we consider
a run of 2-opt with a random initialization, then the probability that we start
the run on a node of the longest path is likely small. If most paths are much shorter,
then this can provide a better explanation for the practical running time of 2-opt
than only studying the longest path.

Structural results on the transition graph may in addition have implications
for the running time of metaheuristics. In particular, 2-opt is often used
as the basis of simulated annealing, a physics-inspired metaheuristic~\cite[Chapter 8]{aartsLocalSearchCombinatorial2003}.
It has long been known
that the structure of the transition graph strongly influences the running time of this algorithm.
Structural parameters of this graph often enter convergence results and running time estimates~\cite{hajekCoolingSchedulesOptimal1988,jerrumLargeCliquesElude1992,nolteNoteFiniteTime2000}.
A recent result by Chen et al.~\cite{chenAlmostLinearPlantedCliques2023} especially illustrates this
point. They showed that the Metropolis process (in essence, simulated annealing at a constant
temperature) is unable to find even very large planted cliques in an Erd\H{o}s-R\'enyi random graph.
Their analysis hinges on several structural results on cliques in such random graphs.

The result by Chen et al., as well as the preceding result by Jerrum~\cite{jerrumLargeCliquesElude1992},
deals with the purely discrete problem of finding cliques. However, simulated annealing is
often applied to \emph{weighted} problems~\cite{aartsSimulatedAnnealingBoltzmann1989},
which yields significant
challenges in understanding the transition graph.
We believe that understanding more about the structure of transition graphs
is key to proving rigorous results on simulated annealing for weighted problems.

\subsection*{Results}

We start by showing that the problem of counting 2-optimal tours is \sharpP-complete,
even on complete weighted graphs. Recall that \sharpP{} is the counting analogue to \NP{},
asking not \emph{whether} a solution exists but \emph{how many} solutions exist.
A formal definition of \sharpP{} is provided in \Cref{sec: counting preliminaries}. 

Our result is in fact slightly stronger. To state it in full, we need the notion of a \emph{path
cover}. A set of paths~$\mathcal{P}$ in a graph~$G$ is a path cover if every vertex of~$G$
is contained in exactly one path of~$\mathcal{P}$. We then have the following.

\begin{restatable}{thm}{countingpathcovers}\label{thm: counting path covers}
    Let~$f_{\text{2-opt}}$ be a function that maps a complete weighted graph on the vertex set~$V$ to
    the number of 2-optimal tours on this graph. Using~$|V|$ calls to~$f_\text{2-opt}$, we can
    compute the number of path covers of size~$\ell$ for each~$1\leq \ell \leq |V|$ in polynomial
    time, using~$f_\text{2-opt}$ as an oracle.
\end{restatable}

This result yields \sharpP{}-hardness of \twoopt{} on complete graphs as a corollary.
This counting problem asks the question:
Given a weighted graph~$G$, how many 2-optimal tours are there on~$G$?
Note that counting the number of Hamiltonian cycles is trivial on complete
graphs, whereas hardness of counting 2-optimal tours in the same setting is
not immediately obvious.

\begin{restatable}{thm}{sharppcomplete}\label{thm: sharp P complete}
    \twoopt{} is \sharpP{}-complete, even on complete graphs.
\end{restatable}

We note that the result remains true for metric TSP instances on complete
graphs, which can be seen to hold by adding a sufficiently large number to every edge weight of
the original instance.

While counting 2-optimal tours on complete graphs is thus likely intractable in general,
we may still wonder about the average case.
In the case where the edge weights are given by
independent uniformly distributed random variables,
we obtain the following upper bound on the number of 2-optimal tours.

\begin{thm}\label{thm:count_2opt}
    Let~$G$ be a complete graph on~$n =2^k+1$ vertices, with edge weights
    drawn independently from~$U[0,1]$ for each edge. Then the expected number of
    2-optimal tours on~$G$ is bounded from above by~$O\left(1.2098^n \sqrt{n!}\right)$.
\end{thm}

In the process of proving \Cref{thm:count_2opt} we obtain a link between
the number of 2-optimal tours and the probability
that all entries of a multivariate normal random vector are positive.
This quantity is also known as the positive orthant probability. To estimate
this probability we prove the following theorem.

\begin{restatable}{thm}{orthantexpectation}\label{thm: orthant expectation}
    Let~$X$ be a multivariate normal vector with zero mean and covariance matrix~$\Sigma$. The positive orthant probability~$\prob(\real_+^d) = \prob(X \in \real_+^d)$ satisfies
    \begin{align*}
        \prob(\real_+^d) \leq
            \frac{\exp\left(\frac{1}{2}\sum_{i=1}^d \Sigma_{ii}^{-1} 
                \left(\expect\left[X_i^2 \,\Big\lvert\, \real_+^d\right]\right)\right)}
                {2^{d-1}e^{d/2}}.
    \end{align*}
    In particular, if the diagonal elements of~$\Sigma^{-1}$ are each~$\Sigma_{ii}^{-1} = 1$,
    then 
    \[
        \prob(\real_+^d) \leq 2^{-d + 1} e^{-d/2} 
        \exp\left(\frac{1}{2}\expect\left[\|X\|_2^2 \,\Big\lvert\, \real_+^d\right]\right).
    \]
\end{restatable}

\Cref{thm: orthant expectation} makes no assumptions on the covariance matrix~$\Sigma$, and
thus holds for any set of zero-mean multivariate normal variables.
Hence, it may be of independent interest in applications
where bounds on positive orthant probabilities are necessary.

\section{Preliminaries}\label{sec: counting preliminaries}

\subsection{Notation and Definitions}

We start with some notational shorthand. Given a symbol~$a$, we write~$a^m$ for the string consisting of~$m$ copies of~$a$.
Throughout,~$\log(\cdot)$ denotes the logarithm to base 2. 
We denote the positive orthant of~$\real^d$
by~$\real_+^d = \{x \in \real^d \given x_i > 0, \, i \in [d]\}$.
The negative orthant~$\real_-^d$ is defined similarly.

Let~$G = (V(G), E(G))$ be a simple graph.
For~$T$ a tour through~$G$, we call the edges of~$T$ the \emph{tour-edges} and the edges of~$E(G) \setminus T$
the \emph{chord-edges}. A 2-change on~$G$ then removes two tour-edges from~$T$
and adds two chord-edges to~$T$.

For a fixed tour, if two tour-edges are removed then there is only one
choice of chord-edges that yields a new tour. Thus, we can characterize a 2-change
fully by the tour-edges it removes.
If a 2-change removes the tour-edges~$e$ and~$f$ we denote this 2-change by~$S_T(e, f)$,
omitting the subscript~$T$ when the tour is clear from the context.
We say that two 2-changes on~$T$
are \emph{chord-disjoint} if they have no chord-edges in common. 

Given a set~$\S$ of~$2$-changes, we define~$P(\S) = \{\{e, f\} \mid S_T(e, f) \in \S\}$
as the set of pairs of tour-edges that participate in the 2-changes in~$\S$.
For~$e \in T$, we define~$k_e(\S) = |\{p \in P(\S) \mid e \in p\}|$, the number of 2-changes
in~$\S$ in which~$e$ participates.
For each of these quantities, we may omit the argument~$\S$ whenever
the set meant is clear from context.

\subsubsection*{Counting Complexity}

For our complexity results, we state the definitions of \sharpP{} and \sharpP{}-completeness
taken verbatim
from Arora and Barak~\cite{aroraComputationalComplexityModern2009}. 

\begin{definition}
    A function~$f: \{0, 1\}^* \to \mathbb{N}$ is in \sharpP{} if there exists a polynomial~$p: \mathbb{N} \to \mathbb{N}$ and a polynomial-time Turing machine~$M$ such that for every~$x \in \{0, 1\}^*$,
    \[
        f(x) = \left|\left\{ y \in \{0, 1\}^{p(x)} \given M(x, y) = 1\right\}\right|.
    \]
\end{definition}

For completeness, we need to recall the complexity class \FP{}, which
is the functional analogue to \Pol{}. This means that \FP{} is the set of
functions~$f : \{0, 1\}^* \to \{0, 1\}^*$ computable by a polynomial-time deterministic 
Turing machine. Moreover, we denote by~$\text{\FP{}}^f$ the set of
such functions where the Turing machine additionally has access to an oracle for~$f$.

\begin{definition}
    A function~$f$ is \sharpP{}-complete if it belongs to \sharpP{} and every~$g \in \text{\sharpP{}}$ is in~$\text{\FP{}}^f$.
\end{definition}

By this definition, if we can solve some \sharpP{}-complete problem in polynomial time,
then we can solve every problem in \sharpP{} in polynomial time.

\subsection{Multivariate Normal Distribution}

We require some basic facts about multivariate normal distributions. Let~$\mu \in \real^d$ and let~$\Sigma \in \real^{d \times d}$ be symmetric and positive definite. The
multivariate normal distribution~$\mathcal{N}(\mu, \Sigma)$ is defined by the
probability density function
\begin{align}\label{eq: multivariate normal distribution}
    f(x) = \frac{\exp\left(-\frac{1}{2}(x - \mu)^T \Sigma^{-1} (x - \mu)\right)}{(2\pi)^{d/2}\sqrt{\det \Sigma}}
\end{align}
with support~$\real^d$.

Let~$X \sim \mathcal{N}_d(\mu, \Sigma)$. The \emph{positive orthant probability}
of~$X$ is the probability that~$X$ falls in the positive
orthant~$\real_+^d$. When~$\mu = 0$ this quantity
is also referred to simply as the \emph{orthant probability}, without specifiying which orthant,
as each orthant is related by flipping the sign of a set of coordinates.

Orthant probabilities are closely related to the \emph{truncated multivariate normal distribution}.
Let~$a, b \in \real^d$ with~$a_i < b_i$ for each~$i \in [d]$. By abuse of notation
we write~$[a, b]^d = \{x \in \real^d \given a_i \leq x_i \leq b_i, \, i \in [d]\}$.
The multivariate normal distribution
truncated from below by~$a$ and from above by~$b$ is given by
\begin{align}\label{eq: truncated normal}
    f(x; a, b) = \begin{cases}
            \frac{1}{\prob\left(X \in [a, b]^d\right)}
             \cdot \frac{\exp\left(-\frac{1}{2}(x - \mu)^T \Sigma^{-1} (x - \mu)\right)}{(2\pi)^{d/2}\sqrt{\det \Sigma}}, &
                    \text{if~$x \in [a, b]^d$}, \\
            0, & \text{otherwise}.
    \end{cases}
\end{align}
A special case of the truncated normal distribution is the \emph{half-normal distribution}, which
is obtained by setting~$a_i = 0$ and~$b_i = \infty$ for all~$i$, and taking~$\mu = 0$ and~$\Sigma = \diag(\sigma_1^2, \ldots, \sigma_d^2)$.

The moments of the truncated multivariate normal distribution are significantly harder to compute than
those for the non-truncated distribution. Nevertheless, there are some elegant results in the literature.
We need the following result due to Amemiya~\cite{amemiyaMultivariateRegressionSimultaneous1974}.

\begin{thm}[\mbox{\cite[Theorem 1]{amemiyaMultivariateRegressionSimultaneous1974}}]\label{thm: amemiya}
    Let~$X$ be distributed according to a~$d$-dimensional truncated multivariate normal distribution
    with zero mean and covariance matrix~$\Sigma$, with each variable truncated only from
    below at zero. Then for each~$i \in [d]$,
    \[
        \sum_{j=1}^d \Sigma^{-1}_{ij}\expect[X_i X_j] = 1.
    \]
\end{thm}

There are also explicit formulae for the second-order moments of this distribution. We use a formula
by Manjunath and Wilhelm~\cite{b.g.MomentsCalculationDouble2009}, adapted to our purposes.

\begin{thm}[\mbox{\cite[adapted from Equation (16)]{b.g.MomentsCalculationDouble2009}}]\label{thm: manjunath and wilhelm}
    Let~$X$ be distributed according to a~$d$-dimensional truncated multivariate normal distribution
    with zero mean and covariance matrix with entries~$\sigma_{ij}$, with each variable truncated 
    only from below at zero. Then for each~$i \in [d]$,
    \[
        \expect[X_i^2] = \sigma_{ii} + \sum_{k=1}^d \sum_{q \neq k} \sigma_{ik}
            \left(
                \sigma_{iq} - \frac{\sigma_{kq}\sigma_{ik}}{\sigma_{kk}}
            \right)F_{kq}(0, 0),
    \]
    where~$F_{ij}$ denotes the joint marginal distribution of~$(X_i, X_j)$.
\end{thm}

\section{Complexity of Counting 2-Optimal Tours}

We now proceed to show hardness of counting 2-optimal tours on complete graphs.
We start by recalling the notion of a \emph{path cover}.
Given a simple graph~$G = (V, E)$, a path cover~$\mathcal{P}$ of~$G$ is a collection of
vertex-disjoint paths such that every~$v \in V$ belongs to exactly
one path of~$\mathcal{P}$. The \emph{size} of~$\mathcal{P}$ is the number
of paths in~$\mathcal{P}$. Note that a path cover of size 1 is equivalent to a Hamiltonian
path, and that a path cover may contain paths that consist of a single vertex.

From an instance~$G = (V, E)$ of \hampath{} we construct a family of instances~$G_m = (V \cup S, E \cup F)$
of \twoopt{} for~$m \geq |V| + 1$.
First, we add a new set~$S$ of vertices to~$G'$, where~$|S| = m \geq |V| + 1$.
To ensure that the reduction is polynomial-time computable, we also require~$m \leq 2|V|$; the reason for this choice will be apparent shortly.
The edges of~$G_m$ are the edges of~$G$, which are assigned a weight of 0
plus the set of edges~$F$, which contains:
\begin{itemize}
    \item all missing edges between the vertices of~$V$, with weight~$M$, which we call non-edges;
    \item all edges between the vertices of~$S$, with weight~$N$, which we call~$S$-edges;
    \item all edges between the vertices of~$V$ and the vertices of~$S$,
        with weight~$L$, which we call~$(V, S)$-edges.
\end{itemize}
The relationship between the edge weights is as follows: we set~$M \gg N = 2L$.
The precise values of these numbers are not important.
See \Cref{fig:reduction} for a schematic depiction of the reduction.

By this construction,~$G_m$ is a complete graph on~$|V| + m$ vertices.
We claim that by computing the number~$f_\text{2-opt}(G_m)$ of 2-optimal tours on~$G_m$
for~$m \in \{|V|+1,\ldots,2|V|\}$ we can
compute the number of path covers of~$G$ of size~$1 \leq \ell \leq |V|$.
To that end, we first characterize the 2-optimal tours on~$G_m$.

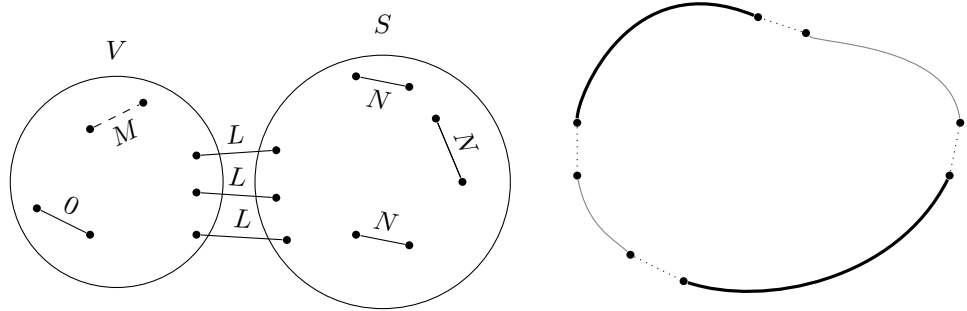
\begin{figure}
    \centering
    \subfigure[The reduction we use to prove \sharpP{}-hardness of
    \twoopt{}. The set~$V$ represents the vertices of the original graph, and~$S$ is 
    the set of~$m$ vertices added in the reduction. Each depicted edge is labelled with its
    weight. Note that in~$G_m$, the non-edges of~$G$ are added as well (represented here by the dashed edge).]{
    \begin{tikzpicture}[scale=0.7]
        \draw (0, 0) circle (2);
        \draw (5, 0) circle (2.4);

        \node at (0, 2.5) {$V$};
        \node at (5, 3) {$S$};

        \vertex (v1) at (-0.5, -1) {};
        \vertex (v2) at (-1.5, -0.5) {};
        \draw (v1) -- node[left,above,sloped] {$0$} ++ (v2);
        
        \vertex (v3) at (-0.5, 1) {};
        \vertex (v4) at (.5, 1.5) {};
        \draw[dashed] (v3) -- node[left,below,sloped] {$M$} ++ (v4);

        \vertex (v5) at (1.5, 0.5) {}; 
        \vertex (v6) at (3, 0.6) {}; 
        \draw (v5) -- node[left,above,sloped] {$L$} ++ (v6);
        \vertex (v5) at (1.5, -0.2) {}; 
        \vertex (v6) at (3, -0.3) {}; 
        \draw (v5) -- node[left,above,sloped] {$L$} ++ (v6);
        \vertex (v5) at (1.5, -1) {}; 
        \vertex (v6) at (3.2, -1.1) {}; 
        \draw (v5) -- node[left,above,sloped] {$L$} ++ (v6);

        \vertex (v7) at (4.5, -1.0) {};
        \vertex (v8) at (5.5, -1.2) {};
        \draw (v7) -- node[left,above,sloped] {$N$} ++ (v8);

        \vertex (v9) at (4.5, 2) {};
        \vertex (v10) at (5.5, 1.8) {};
        \draw (v9) -- node[left,below,sloped] {$N$} ++ (v10);

        \vertex (v11) at (6, 1.2) {};
        \vertex (v12) at (6.5, 0) {};
        \draw (v11) -- (v12);
        \draw (v11) -- node[right,above,sloped] {$N$} ++ (v12);
    \end{tikzpicture}
    }
    \hspace{1em}
    \subfigure[A 2-optimal tour through~$G_m$ consisting of two segments. Solid black
    curves represent paths through~$V$, solid gray lines are paths through~$S$, and dotted lines
    are~$(V, S)$-edges.]{
    \begin{tikzpicture}[scale=0.7]
        \vertex (v1) at (0, 0) {};
        \vertex (v2) at (5, 2) {};
        \draw[very thick] (v1) .. controls (1.5, -0.5) and (4, 0) .. (v2);
        \vertex (v3) at (5.2, 3) {};
        \draw[dotted] (v2) -- (v3);

        \vertex (v4) at (2.3, 4.7) {};
        \draw[gray] (v3) .. controls (5., 4.5) and (2.5, 4.5) .. (v4);

        \vertex (v5) at (1.4, 5) {};
        \draw[dotted] (v4) -- (v5);

        \vertex (v6) at (-2, 3) {};
        \draw[very thick] (v5) .. controls (-1, 6) and (-2, 3.5) .. (v6);

        \vertex (v7) at (-2, 2) {};
        \draw[dotted] (v6) -- (v7);

        \vertex (v8) at (-1, 0.5) {};
        \draw[gray] (v7) .. controls (-1.8, 1) and (-1.3, 0.8) .. (v8);
        \draw[dotted] (v8) -- (v1);

    \end{tikzpicture}
    }
    \caption{Schematic depiction of the reduction we use to prove \sharpP{}-hardness of \twoopt{},
    and of a 2-optimal tour in the image instance.}
    \label{fig:reduction}
\end{figure}

\begin{lemma}\label{lemma:2opt non-edges}
    A tour~$T$ through~$G_m$ is 2-optimal if and only if it contains no non-edges.
\end{lemma}

\begin{proof}
    To begin, we note that since $|S| > |V|$, any tour must contain at least one~$S$-edge. 
    Now suppose~$T$ contains a non-edge~$uv \notin E$, and let~$ab$ be an~$S$-edge in~$T$.
    Assume these vertices are traversed in the written order on~$T$. 
    Then we replace~$uv$ and~$ab$ by~$ua$ and~$vb$, a 2-change which decreases the tour length
    by~$M + N - 2L > 0$. Hence,~$T$ is not 2-optimal, proving one direction.

    For the other direction, suppose~$T$ contains no non-edges.
    Since~$M$ is much larger than~$N$ or~$L$, an improving 2-change cannot
    add any non-edges to the tour.
    Any 2-changes on~$T$ must then involve only edges of~$G$,~$S$-edges and~$(V, S)$-edges. 
    We go case by case,
    considering the possible types of the edges removed by any 2-change on~$T$.

    \begin{description}
        \item[\boldmath Two~$S$-edges.] Since the vertices involved are all vertices of~$S$
            the added edges are also~$S$-edges, and hence the 2-change does not change the length
            of the tour at all and is therefore not feasible.
    \item[\boldmath Two edges of~$G$.] Since these edges have zero weight the
            tour length cannot decrease, and the 2-change is not feasible.
    \item[\boldmath Two~$(V, S)$-edges.] There are two possibilities
            for the added edges. In one case we obtain one~$S$-edge and one edge of~$G$,
            for an improvement of~$2L - N = 0$. In the other case we obtain
            again two~$(V, S)$-edges, again for no improvement.
    \item[\boldmath One~$S$-edge~$ab$ and one edge~$uv$ of~$G$.] The added edges
            are~$au$ and~$bv$, both of which are~$(V, S)$-edges and thus of weight~$L$. The tour
            is thus shortened by~$N - 2L = 0$, and the 2-change yields no improvement.
    \item[\boldmath One~$(V, S)$-edge $au$ and one edge~$vw$ of~$G$.]
            The added edges are~$av$ and~$uw$. The added edges are again a~$(V, S)$-edge
            and an edge of~$G$, yielding no improvement.
    \item[\boldmath One~$(V, S)$-edge~$au$ and one~$S$-edge~$bc$.] 
            The added edges consist of one~$(V, S)$-edge and one~$S$-edge, and there
            is no change in tour length.
    \end{description}

    These cases cover all possibilities, and hence a tour that contains no non-edges is
    2-optimal, concluding the proof.
\end{proof}

Let~$T$ be a 2-optimal tour through~$G_m$. If we remove from~$T$ all edges incident
to~$S$, then we obtain a path cover of~$G$. We call the size of the resulting path cover
the number of \emph{segments} of~$T$. On the other hand, from any path cover
of~$G$, we can construct many 2-optimal tours by connecting the paths in the cover
using vertices of~$S$.

More formally, we say that a path cover~$\mathcal{P}$ \emph{corresponds} to~$T$ if we obtain~$\mathcal{P}$ by removing the edges of~$T$ incident to~$S$.
Note that two distinct path covers of~$G$ correspond to two disjoint sets of 2-optimal tours
through~$G_m$, and every 2-optimal tour corresponds to exactly one path cover.
The following lemma counts the number of 2-optimal tours that correspond to a single
path cover of size~$\ell$.

\begin{lemma}\label{lemma: pc to tours}
    A path cover of size~$\ell$ corresponds to exactly~$\frac{2^{\ell-1} \cdot m! \cdot (m-1)!}{(m-\ell)!}$ 2-optimal tours in~$G_m$
    with~$\ell$ segments.
\end{lemma}

\begin{proof}
    By \Cref{lemma:2opt non-edges}, the 2-optimal tours are exactly those
    tours which contain no non-edges. Given a path cover~$\mathcal{P}$ of size~$\ell$,
    we can then construct a 2-optimal tour in~$G_m$ as follows. 
    
    Any tour must visit all vertices in~$S$ in some order. Hence, we first fix a tour~$T_S$ through~$S$; there are~$\frac{1}{2}(m-1)!$ such choices. Next, we break~$\ell$ edges of~$T_S$;
    there are~$\binom{m}{\ell}$ choices of which edges to break.

    We must now insert the~$\ell$ paths of~$\mathcal{P}$ in place of these broken edges,
    reconnecting the endpoints of each path to the endpoints of the broken
    edge it replaces.
    Note that whenever we perform this operation, there are two possible ways to
    connect the path. Moreover, there are~$\ell!$ ways to match an endpoint
    to a broken edge, for a total of~$2^\ell \cdot \ell!$ ways to insert the paths.
    
    Putting the pieces together, we thus have
    \[
        \frac{(m-1)!}{2} \cdot \binom{m}{\ell} \cdot 2^{\ell} \cdot \ell !
            = \frac{2^{\ell-1}\cdot m! \cdot (m-1)!}{(m-\ell)!}
    \]
    ways to construct a 2-optimal tour through~$G_m$.
\end{proof}

Let~$c(\ell, m) = \frac{2^{\ell-1}\cdot m! \cdot (m-1)!}{(m-\ell)!}$, and 
let~$C \in \mathbb{Z}^{|V| \times |V|}$ be a matrix with entries~$C_{ij} = c(i, j+|V|)$.
Recall that~$\ell$ runs from~$1$ to~$|V|$ and~$m$ runs from~$|V|+1$ to~$2|V|$. 

\begin{lemma}\label{lemma: matrix rank}
    The matrix~$C$ defined above has full rank.
\end{lemma}

\begin{proof}
    To start, we write out the entries of~$C$ explicitly. Letting~$n = |V|$,
    \[
        C_{ij} = \frac{2^{i-1} \cdot (n+j)!\cdot(n+j-1)!}{(n+j-i)!}.
    \]
    Scaling any row or column of a matrix uniformly does
    not change its rank. Hence, we multiply column~$j$ by~$\frac{1}{(n+j)!\cdot(n+j-1)!}$,
    and subsequently multiply row~$i$ by~$1/2^{i-1}$. Finally, interchanging any two rows
    also does not change the rank of the matrix, and hence we also mirror the rows
    of the matrix. The resulting matrix~$C'$ has entries
    \[
        C'_{ij} = \frac{1}{(i+j-1)!}.
    \]
    To show that this matrix has full rank, we first recall that a square matrix has
    full rank if and only if its determinant is nonzero. While we could in principle
    compute the determinant exactly, this is tedious and not necessary for our purposes. Hence,
    we use a concise argument from analysis to only show~$\det C' \neq 0$~\cite{yuanAnswerDeterminantMatrix2022},
    reproduced here for the sake of completeness. The proof uses the notion of
    the \emph{Wronskian}~$W(x)$ of a set of functions~$f_j(x)$, which is the determinant of the matrix
    with~$(i,j)$-th element~$f_j^{(i)}(x)$. If the functions are analytic, then~$W(x)$ vanishes
    identically if and only if the functions are linearly dependent \cite{bocherCertainCasesWhich1901}.

    We observe that the determinant of~$C'$ is, up to a sign, the Wronskian
    of the functions~$f_j(x) = \frac{x^{n+j-1}}{(n+j-1)!}$ evaluated at~$x = 1$.
    Since these functions are polynomials of different degrees, they are linearly independent
    and hence from analyticity of the polynomials it follows that
    the Wronskian does not vanish identically. By factoring out all the powers of~$x$ from the rows and columns, we see that the Wronskian is a power of~$x$
    multiplied by its value at~$x = 1$, which is~$\det C'$ up to a sign. As the Wronskian
    does not vanish identically, we must have~$\det C' \neq 0$.
\end{proof}

Now we proceed to our main algorithmic result, from which \sharpP{}-completeness of
\twoopt{} follows as a corollary.

\countingpathcovers*

\begin{proof}
    For a graph~$G$, let~$G_m$ be the complete weighted graph resulting from the reduction described
    above. Let~$a_{\ell}(G)$ be the number of path covers
    of size~$\ell$, and let~$b_{\ell}(G_m)$ be the number of 2-optimal tours on~$G_m$ consisting of~$\ell$ segments. From \Cref{lemma: pc to tours}, we have
    \[
        b_{\ell}(G_m) = c(\ell, m)a_\ell(G).
    \]
    We have~$f_\text{2-opt}(G_m) = \sum_{\ell=1}^{|V|} b_{\ell}(G_m)
        =\sum_{\ell=1}^{|V|} c(\ell, m)a_\ell(G)$.

    We aim to compute the numbers~$a_{\ell}(G)$ for each~$\ell \in [|V|]$. Let~$a = (a_{\ell}(G))_{\ell \in [|V|]}$, and let~$b = (f_\text{2-opt}(G_m))_{m = |V|+1}^{2|V|}$. Then the above
    yields the matrix equation~$b = C^Ta$.
    
    By \Cref{lemma: matrix rank}, the matrix~$C$ has full rank, and is hence invertible.
    While~$C$ is rather ill-conditioned, the elements of~$C$ can be encoded using a number of bits 
    polynomial in~$|V| + |E|$. Thus, after making~$|V|$ calls to~$f_\text{2-opt}$ to
    compute the vector~$b$, we can compute~$a$ in polynomial time. The entries~$a_\ell$
    of~$a$ are, by construction, exactly the number of path covers of size~$\ell$ of~$G$.
\end{proof}

\sharppcomplete*

\begin{proof}
    Membership in \sharpP{} is obvious, since the problem of verifying 2-optimality
    of a tour is in \Pol{}.
    For hardness, we rely on the \sharpP{}-hardness of \hampath{}, which was shown by
    Valiant~\cite{valiantComplexityEnumerationReliability1979}. Note that
    the number of Hamiltonian paths through a graph~$G$ is exactly the number of
    path covers of size~$\ell = 1$. 
    By \Cref{thm: counting path covers}, given oracle access to~$f_\text{2-opt}$, we
    can solve \hampath{} -- and thus every problem in \sharpP{} -- in polynomial time,
    and hence \twoopt{} is \sharpP{}-complete when restricted to complete graphs.
\end{proof}

\section{Counting 2-Optima in Random Instances}

While counting 2-optimal tours is in general hard on complete graphs,
we can still provide some results for special cases. In this section, we restrict
our attention to complete graphs on~$n = 2^k + 1$ vertices for some integer~$k$.
The weights of the edges of our graphs are drawn independently from the uniform
distribution on~$[0, 1]$.

The strategy we use is to bound the probability that a given tour is 2-optimal. Since
we assume the input graph is complete, this probability is the same for all tours.
It then suffices to multiply this probability by the total number of tours,
which is~$\frac{1}{2}(n-1)!$.

To bound the probability of a tour~$T$ being 2-optimal
we find a large set~$\S$ of mutually chord-disjoint 2-changes on~$T$.
We then apply a general result that links the probability of~$2$-optimality
of~$T$ to how often each edge of~$T$ is used in~$\S$ (\Cref{lemma:2-optimal chord-disjoint}).
Details of the construction of this set are given in \Cref{sec: construct S}.

\subsection{Orthant Probabilities and 2-Optimality}\label{sec: orthant and 2opt}

Let~$\S$ be a set of chord-disjoint 2-changes on a tour~$T$ on the complete
graph on~$n$ vertices, and define for~$x \in \real^E$
\[
    g_{\S}(x) = \exp\left(-\sum_{\{e, f\} \in P(\S)} \frac{x_e x_f}{\sqrt{k_e(\S) k_f(\S)}}\right).
\]
Now let~$X = (X_e)_{e \in T, k_e(\S) > 0}$ be a sequence of independent half-normal distributed
random variables with unit variance. We define the function
\[
    \G(\S) = \expect[g_{\S}(X)].
\]
Observe that~$\G$ is solely a function of~$\S$.

\begin{lemma}\label{lemma:2-optimal chord-disjoint}
    Let~$G$ be a complete graph on~$n$ vertices, with each edge
    of~$G$ independently assigned a weight drawn from~$U[0, 1]$. Let~$T$
    be any tour through~$G$. Let~$\mathcal{S}$ be a set of chord-disjoint
    2-changes on~$T$. Then
    \[
        \prob(\text{$T$ is 2-optimal})
            \leq 
            \G(\S) \prod_{\substack{e \in T \\ k_e(\S) > 0}} \sqrt{\frac{\pi}{2k_e(\S)}}.
    \]
\end{lemma}

Before moving to a proof, we need a technical lemma.

\begin{lemma}\label{lemma:gaussian bound}
    For~$1 \leq x \leq 2$, it holds that~$\frac{1}{2}(2 - x)^2 \leq e^{-x^2/2}$.
\end{lemma}

\begin{proof}
    First, observe that the inequality holds for~$x = 1$, as~$\frac{1}{2} < 1/\sqrt{e} \approx 0.6$.
    Next, note that~$e^{-x^2/2}$ is strictly convex on~$(1, \infty)$.
    Thus, this function lies entirely above its tangent line
    at~$x = 1$ for all~$x > 1$. The tangent line is given by~$x \mapsto 2e^{-1/2} - e^{-1/2}x$. It then suffices to show that~$\frac{1}{2}(2-x)^2 \leq 2e^{-1/2} - e^{-1/2}x$ for~$x \in (1, 2]$.
    Indeed,
    \[
        \frac{1}{2}(2-x)^2 \leq 2e^{-1/2} - e^{-1/2}x \iff
            2\left(1 - \frac{1}{\sqrt{e}}\right) - \left(2 - \frac{1}{\sqrt{e}}\right)x
                 + \frac{1}{2}x^2 \leq 0,
    \]
    which holds for~$x \in [2-2/\sqrt{e}, 2] \supset (1, 2]$.
\end{proof}

\begin{proof}[Proof of \Cref{lemma:2-optimal chord-disjoint}]
    Let~$S$ be a 2-change on~$T$ and write~$\Delta(S)$ for the improvement in tour length
    due to~$S$. Then~$T$ is 2-optimal if and only if for all possible 2-changes~$S$
    on~$T$ we have~$\Delta(S) \leq 0$, and so
    \[
        \prob(\text{$T$ is 2-optimal}) \leq \prob\left(
            \bigwedge_{S \in \mathcal{S}} \Delta(S) \leq 0
        \right).
    \]
    If we condition on the values of the weights on the edges of~$T$, then,
    since the 2-changes in~$\mathcal{S}$ are all chord-disjoint, the events~$\{\Delta(S) \leq 0\}_{S \in \mathcal{S}}$
    are independent subject to this conditioning. Thus,
    \[
        \prob(\text{$T$ is 2-optimal} \given w(e) = s_e, \, e \in T) \leq \prod_{S \in \mathcal{S}}\prob\left(
            \Delta(S) \leq 0 \given w(e) = s_e,\, e \in T
        \right).
    \]

    Consider a 2-change~$S$ that involves the tour-edges~$e_i$ and~$e_j$. Then
    \[
        \prob(\Delta(S) \leq 0 \given w(e_i) = s_i,\, i \in [n])
            = \prob\left(X + Y \geq s_i + s_j\right),
    \]
    where~$X, Y \sim U[0,1]$. We can directly compute this latter probability, yielding
    \[
        \prob\left(X + Y \geq s_i + s_j\right) = \int_{s_i + s_j}^2 f_{X + Y}(x)\dd x,
    \]
    where the integrand is the density of~$X+Y$, given by
    \[
        f_{X + Y}(x) = \begin{cases}
            x, & \text{if~$0 \leq x \leq 1$}, \\
            2 - x, & \text{if~$1 < x \leq 2$}, \\
            0, & \text{otherwise}.
        \end{cases}
    \]
    Integrating this density, we obtain
    \[
        \prob(X + Y \geq s_i + s_j) = \begin{cases}
            1, & \text{if } s_i + s_j \leq 0, \\
            1 - \frac{1}{2}(s_i + s_j)^2, & \text{if } s_i + s_j \leq 1, \\ 
            \frac{1}{2}(2 - (s_i + s_j))^2, & \text{if } 1 < s_i + s_j \leq 2, \\
            0, & \text{if } s_i + s_j > 2.
        \end{cases}
    \]
    On~$[0, 2]$, this function is bounded from above by 
    \[
        e^{-\frac{(s_i + s_j)^2}{2}} = e^{-s_i^2/2}e^{-s_j^2/2} e^{-s_i s_j}.
    \]
    For~$0\leq s_i + s_j \leq 1$, this follows from the standard inequality~$1 + x \leq e^x$,
    while for~$1 < s_i + s_j \leq 2$ we use \Cref{lemma:gaussian bound}.

    Using this upper bound, we can write
    \[
        \prob(\text{$T$ is 2-optimal} \given w(e) = s_e, \, e \in T)
            \leq \prod_{\{e, f\} \in P(\S)} e^{-s_e^2/2}e^{-s_f^2/2}e^{-s_e s_f}.
    \]
    Note that for a given edge~$e \in T$, the factor factor~$e^{-s_e^2/2}$ appears~$k_e$ times. Hence, 
    \[
        \prob(\text{$T$ is 2-optimal} \given w(e) = s_e, \, e \in T)
            \leq \prod_{e \in T} e^{-k_e s_e^2/2} \cdot \prod_{\{e, f\} \in P(\S)}
                e^{-s_e s_f}.
    \]

    We now get rid of the conditioning,
    \begin{align*}
        \prob(\text{$T$ is 2-optimal}) &\leq \int_0^1 \cdots \int_0^1 
                    \exp\left(-\frac{1}{2}\sum_{\substack{e \in T \\ k_e > 0}} k_e s_e^2\right)
                    \exp\left(-\sum_{\{e, f\} \in P(\S)} s_e s_f\right)
                \prod_{\substack{e \in T \\ k_e > 0}}\dd s_e.
    \end{align*}
    We perform a change of variables: let~$x_e = \sqrt{k_e} s_e$. Next, we let the
    upper limit of each integral go to infinity; it is easily verified that this
    leads to a negligible loss in the upper bound. We find
    \begin{align*}
        \prob(\text{$T$ is 2-optimal}) &\leq  \prod_{\substack{e \in T \\ k_e > 0}} \frac{1}{\sqrt{k_e}}
        \int_{\real_+^d}
                    \exp\left(-\frac{1}{2}\sum_{\substack{e \in T \\ k_e > 0}} x_e^2\right)
                    \exp\left(-\sum_{\{e, f\} \in P(\S)} \frac{x_e x_f}{\sqrt{k_e k_f}}\right)
                \dd x,
    \end{align*}
    where~$\dd x = \prod_{\substack{e\in T \\ k_e > 0}} \dd x_e$.
    We recognize in this expression the function~$g_{\S}$, as well as
    the un-normalized probability density function of the half-normal distribution. Adding
    in the appropriate normalization factor of~$\sqrt{2/\pi}$ for each variable then yields the claim.
\end{proof}

Bounding~$\prob(\text{$T$ is 2-optimal})$ now involves two steps. First, we must construct
a set~$\S$ of chord-disjoint 2-changes such that each edge of~$T$ is used in many 2-changes in~$\S$.
Second, given this set, we must bound~$\G(\S)$.

We remark now that~$\G(\S)$ is trivially bounded from above by 1.
However, this leaves a factor  of about~$(\pi/2)^{n/2}$. 
Although this factor is small compared to~$\prod_{e \in T, \,k_e > 0} \frac{1}{\sqrt{k_e(\S)}}$ for the
set~$\S$ we construct,  leaving it in is somewhat
unsatisfactory. We thus make an attempt to prove a stronger bound for~$\G$.

Computing~$\G(\S)$ directly is unfeasible. It helps to recast it in terms of
a positive orthant probability.

\begin{lemma}\label{lemma: rewrite to orthant probability}
    For a set~$\S$ of chord-disjoint 2-changes on a tour~$T$, let~$k$ denote the number of
    tour-edges with~$k_e(\S) > 0$. Label these edges of~$T$ arbitrarily from~$1$ to~$k$.
    For any~$I \subseteq [k]$, the function~$\G(\S)$ is bounded from above
    by~$2^{|I|} \cdot \sqrt{\det{\Sigma}} \cdot \prob\left(\bigwedge_{i \in I} Z_i > 0\right)$, where~$(Z_i)_{i \in I}$
    is distributed according to a multivariate normal distribution with mean 0 and
    inverse covariance matrix~$\Sigma^{-1} \in \mathbb{R}^{|I| \times |I|}$ with entries indexed by~$I$,
    \[
        \Sigma^{-1}_{ij} = \begin{cases}
            1, & \text{ if~$i = j$}, \\
            s_{ij}, & \text{ otherwise},
        \end{cases}
    \]
    where~$s_{ij} \leq 1/\sqrt{k_i k_j}$.
\end{lemma}

\begin{proof}
    Since~$g_{\S}(x) \leq 1$ and~$g_{\S}$ is decreasing in each variable,
    we can bound~$g_{\S}$ from above uniformly by setting the coefficients of any subset of the variables to zero. 
    Setting these to zero for
    all variables outside~$I$ then leaves
    \begin{align*}
        \G(\S) \leq \left(\frac{2}{\pi}\right)^{|I|/2} \int_{\real_+^{|I|}}
            \exp\left(
                -\frac{1}{2}\sum_{i \in I} x_i^2 
            \right)\exp\left(-\sum_{i \in P_I(\S)} \frac{x_i x_j}{\sqrt{k_i k_j}}\right) \prod_{i \in I} \dd x_i,
    \end{align*}
    where~$P_I(\S)$ is the same as~$P(\S)$, except we keep only pairs with both elements
    in~$I$.

    Observe that~$e^{-x_i x_j / \sqrt{k_i k_j}} \leq e^{-s_{ij} x_i x_j}$, allowing
    us to replace the coefficients of these products. We now only need to insert
    the appropriate normalization factor for a multivariate normal distribution. Comparing
    the resulting expression with
    \Cref{eq: multivariate normal distribution} completes the proof.
\end{proof}

Still, computing the positive orthant probability directly is rather difficult.
Explicit formulas are known for low-dimensional cases, as well as general recursive
formulas~\cite{abrahamsonOrthantProbabilitiesQuadrivariate1964,chengOrthantProbabilitiesFour1969,davidNoteEvaluationMultivariate1953},
but none of these are particularly helpful in bounding~$\G(\S)$.
As we only need a nontrivial upper bound, some simplifications are possible.
In \Cref{thm: orthant expectation} (restated below), we show that it suffices to bound
the expected squared norm of a multivariate normal vector. 

In the proof of \Cref{thm: orthant expectation}, we need another technical
lemma.

\begin{lemma}\label{lemma: string pairs}
    Let~$S$ be a random string uniformly chosen from~$\{1, -1\}^{d} \setminus \{1^d, -1^d\}$
    for~$d \geq 2$.
    Then we have~${\prob(S_i + S_j = 0) \geq \frac{1}{2}}$ for any distinct~$i, j \in [d]$.
\end{lemma}

\begin{proof}
    The case~$d = 2$ yields~$\prob(S_1 + S_2 = 0) = 1 \geq \frac{1}{2}$; hence, we assume~$d > 2$
    in the remainder.
    Note next that by symmetry, it suffices to consider~$i = 1$,~$j = 2$.
    The event~$S_1 + S_2 = 0$ occurs if and only if~$(S_1, S_2)$ is either~$(1, -1)$
    or~$(-1, 1)$.

    We use the principle of deferred decisions. Suppose that the remaining~$d-2$ variables have
    been fixed. If the string formed by these remaining variables is equal to neither~$1^{d-2}$ nor~$-1^{d-2}$, then the outcome of drawing~$(S_1, S_2)$ is unconstrained.
    There are four outcomes and by symmetry each outcome has equal probability, yielding the claim.
    If the remaining string is~$1^{d-2}$, then there are only three
    possible outcomes for~$(S_1, S_2)$, each with equal probability. Two of these outcomes
    satisfy~$S_1 + S_2 = 0$; hence in this case~$\prob(S_1 + S_2 = 0) = \frac{2}{3} \geq \frac{1}{2}$. The case where the remaining
    string is~$-1^{d-2}$ is identical.
\end{proof}

\orthantexpectation*

\begin{proof}
    We prove the first claim, as the second claim follows trivially.
    For~$s \in \{-1, 1\}^d$, let~$\real_s^d$ be the orthant
    corresponding to~$s$, given by the points~$x \in \real^d$ satisfying
    \[
        s_i x_i > 0, \quad i \in [d].
    \]
    With this notation, the positive and negative orthants are given by~$\real_{\pm}^d = \real^d_{\pm 1^d}$.
    
    Since~$\real^d = \bigcup_{s \in \{\pm 1\}^d} \real^d_s$ and the orthants are mutually
    disjoint,~$\sum_{s \in \{\pm 1\}^d} \prob(\real_s) = 1$.
    By symmetry,~$\prob(\real_+^d) = \prob(\real_-^d)$, and so
    \begin{align}\label{eq: orthant probability rearranged}
        \prob(\real_+^d) = \frac{1}{2} 
            - \frac{1}{2} \sum_{\substack{s \in \{\pm 1\}^d \\ s \neq \pm 1^d}}
                \prob(\real_s^d).
    \end{align}

    Given~$s \in \{\pm 1\}^d$, define the linear
    transformation~$R_s(x) = (s_i x_i)_{i \in [d]}$.
    Any~$x \in \real_s^d$ can then be written as~$x = R_s(x')$ for some~$x' \in \real_+^d$. 
    Let~$\Sigma_s^{-1}$ denote the matrix with the same entries as~$\Sigma^{-1}$,
    but with zeroes on the diagonal and zeroes for any entries~$(i, j)$ with~$s_i + s_j \neq 0$.
    Now note that for~$x \in \real_s^d$,
    \[
        \frac{1}{2}x^T \Sigma^{-1} x = \frac{1}{2}\sum_{i=1}^{d}\sum_{j=1}^d \Sigma_{ij}^{-1}x_i x_j
            = \frac{1}{2}\sum_{i=1}^d\sum_{j=1}^d \Sigma^{-1}_{ij} x'_i x'_j 
            - \sum_{i = 1}^d \sum_{j=1}^d \Sigma_{s, ij}^{-1} x'_i x'_j,
    \]
    since only the terms satisfying~$s_i + s_j = 0$ change sign under~$R_s$.
    Thus, we have
    \begin{align*}
        \prob(\real_s^d)
            &= \frac{1}{(2\pi)^{d/2}\sqrt{\det \Sigma}} \int_{\real_+^d}
                e^{-\frac{1}{2}x^T \Sigma^{-1} x}
                \exp\left(\sum_{i=1}^d\sum_{j=1}^d
                    \Sigma^{-1}_{s, ij} x_i x_j\right)
                \dd x \\
            &= \expect\left[
                \exp\left(\sum_{i=1}^d \sum_{j=1}^d
                    \Sigma^{-1}_{s,ij} X_i X_j\right)
                    \,\Biggl\vert\, \real_+^d
            \right] \prob(\real_+^d).
    \end{align*}
    Inserting this into \Cref{eq: orthant probability rearranged} and
    rearranging yields
    \begin{align*}
        \prob(\real_+^d) &= \left(
            2 + \sum_{\substack{s \in \{1, -1\}^d \\ s \neq \pm 1^d}}
                \expect\left[
                    \exp\left(\sum_{i=1}^d\sum_{j=1}^d \Sigma_{ij}^{-1} X_i X_j\right)
                \,\Biggl\vert\, \real_+^d\right]
        \right)^{-1}.
    \end{align*}
    
    Since~$\exp(\cdot)$ is a convex function, we bound the denominator from
    below by moving the expectation operator inside the~$\exp(\cdot)$ using Jensen's inequality.
    To shorten our notation, we replace the variables~$X_i$ by~$Z_i$ distributed according
    to a multivariate normal distribution truncated
    from below at zero. We abbreviate the expectation with respect to~$Z$ by~$\expect_Z[\cdot]$.
    Then (discarding the~$2$ in the denominator above)
    \[
        \prob(\real_+^d) \leq
            \left(\sum_{\substack{s \in \{\pm 1\}^d \\ s \neq \pm 1^d}}\exp\left(
            \expect_Z\left[
            \sum_{i=1}^d\sum_{j=1}^d \Sigma_{s,ij}^{-1} Z_i Z_j\right]
            \right)\right)^{-1}.
    \]
    Let~$S$ be a random string drawn uniformly from~$\{\pm 1\}^d \setminus \{\pm 1^d\}$.
    Observe that the sum in brackets above is the same as
    \[
        \left(2^{d}-2\right)\expect_S\left[\exp\left(
            \expect_Z\left[
            \sum_{i=1}^d \sum_{j=1}^d\Sigma_{S,ij}^{-1} Z_i Z_j\right]
            \right)\right].
    \]
    Another application of Jensen's inequality moves the~$\expect_S[\cdot]$ into the~$\exp(\cdot)$. To simplify the expression yet further, we bound~$2^d - 2 \geq 2^{d-1}$ for~$d \geq 2$.

    Let~$Y_{ij}(S)$ be an indicator random variable taking a value of 1
    if~$\Sigma^{-1}_{s, ij} > 0$ and 0 otherwise. Then~$\Sigma^{-1}_{S, ij} = Y_{ij}(S)\Sigma^{-1}_{ij}$. Note that~$Y_{ij}(S) = 1$ if and only if~$S_i + S_j = 0$. We
    then use \Cref{lemma: string pairs} to obtain~$\expect_S[Y_{ij}(S)] = \prob(Y_{ij}(S) = 1) \geq \frac{1}{2}$. Using~$Y_{ij}(S)$, we further rewrite our bound to
    \begin{align*}
        \expect_S\left[\expect_Z\left[
            \sum_{i=1}^d\sum_{j=1}^d \Sigma_{s,ij}^{-1} Z_i Z_j\right]\right]
            &= 
           \sum_{i=1}^d\sum_{\substack{j=1 \\ j \neq i}}^d
                \expect_S\left[
                    \expect_Z\left[
                        Y_{ij}(S) \Sigma_{ij}^{-1} Z_i Z_j
                        \right]
                    \right] \\
            &= 
            \sum_{i=1}^d\sum_{\substack{j=1 \\ j \neq i}}^d
                \expect_S[Y_{ij}(S)]\left[\expect_Z\left[
                \Sigma_{ij}^{-1} Z_i Z_j\right]\right] \\
            &\geq
            \frac{1}{2}\sum_{i=1}^d\sum_{\substack{j=1 \\ j \neq i}}^d
                \Sigma_{ij}^{-1}
                \expect_Z\left[
                    Z_i Z_j
                \right].
    \end{align*}
    The second equality follows from the independence of~$\{Z_i\}_{i=1}^d$ and~$Y_{ij}(S)$. 

    For the final step we use \Cref{thm: amemiya}, from which we conclude
    \begin{align*}
        \frac{1}{2}\sum_{i=1}^d\sum_{\substack{j=1 \\ j \neq i}}^d \Sigma_{ij}^{-1}
            \expect_Z\left[Z_i Z_j\right]
            &= \frac{1}{2}\sum_{i=1}^d \sum_{j=1}^d \Sigma_{ij}^{-1}\expect_Z\left[ Z_j Z_j\right]
                 - \frac{1}{2}\sum_{i=1}^d \Sigma_{ii}^{-1}\expect_Z[Z_i^2] \\
                &= \frac{d}{2} - \frac{1}{2}\sum_{i=1}^d \Sigma_{ii}^{-1} \expect_Z[Z_i^2].
    \end{align*}
    Putting the pieces together now yields the claim.
\end{proof}

\subsection{Finding Chord-Disjoint 2-Changes}\label{sec: construct S}

To proceed, we need to construct an appropriate set~$\S$ of chord-disjoint 2-changes.
We provide an explicit construction of such a set.
In doing so, we need
the following four lemmas to help us characterize when a pair of 2-changes is chord-disjoint.

It is convenient for the remainder of the section to assign an orientation
to~$T$. Pick an arbitrary vertex of~$T$, and walk along~$T$ in an arbitrary
direction. We order the edges of~$T$ according to the order we encounter
them in, labelling the first edge~$e_1$, the second~$e_2$, and so on.
Moreover, we consider these edges \emph{directed}: If~$e$ is incident
to~$u$ and~$v$ and~$u$ is encountered before~$v$, then we write~$e = uv$.

\begin{lemma}\label{lemma:share one chord disjoint}
    If two 2-changes share exactly one tour-edge, then they are chord-disjoint.
\end{lemma}

\begin{proof}
    Let~$S_1 = S(e, f_1)$ and~$S_2 = S(e, f_2)$ be two 2-changes sharing a single tour-edge~$e = u_1 u_2$.
    Let~$f_1 = v_1 v_2$ and~$f_2 = v_3 v_4$. Since~$f_1 \neq f_2$, these edges can share at most one endpoint.
    Moreover, by the orientation of~$T$, we can only have~$v_2 = v_3$ or~$v_1 = v_4$
    
    The chord-edges involved in~$S_1$ are~$u_1 v_1$ and~$u_2 v_2$,
    while the chord-edges involved in~$S_2$ are~$u_1 v_3$ and~$u_2 v_4$.
    Thus,~$S_1$ and~$S_2$ only share a chord-edge if~$v_1 = v_3$ or~$v_2 = v_4$, which is not possible.
\end{proof}

\begin{lemma}\label{lemma:at most one endpoint}
    If two 2-changes have at most one endpoint in common, then they are chord-disjoint.
\end{lemma}

\begin{proof}
    Let~$S_1$ and~$S_2$ be two 2-changes. Note that by assumption,~$S_1$
    and~$S_2$ cannot share any tour-edges, since then they would have two endpoints in common.
    If the 2-changes have no endpoints in common, then the lemma is obviously true.

    Assume then that the 2-changes have one endpoint in common. Let~$e_1$ and~$f_1$ be removed by~$S_1$, and~$e_1$ and~$f_2$ be removed by~$S_2$.
    Without loss of generality, assume~$e_1$ and~$e_2$ have an endpoint in common. 
    We write~$e_1 = v_1 v_2$ and~$e_2 = v_2 v_3$. Let~$f_1 = u_1 u_2$ and~$f_2 = u_3 u_4$.
    Note that all of these vertices with different labels are distinct.

    The chord-edges added by~$S_1$ are then~$v_1 u_1$ and~$v_2 u_2$, while those
    added by~$S_2$ are~$v_2u_3$ and~$v_3 u_4$. Observe that these are four distinct
    edges, concluding the proof.
\end{proof}

\begin{lemma}\label{lemma:successive chord disjointness}
    Let~$T_1$ and~$T_2$ be two successive sub-paths of a tour~$T$, both containing
    an even number of edges. Any two 2-changes which are each
    formed by removing one edge in~$T_1$ and one
    edge in an even position along~$T_2$ are chord-disjoint.
\end{lemma}

\begin{proof}
    Let~$e_1, e_2 \in T_1$ and~$f_1, f_2 \in T_2$. Consider two distinct 2-changes~$S_1 = S(e_1, f_1)$ and~$S_2 = S(e_2, f_2)$.
    Note that if~$e_1 \neq e_2$ and~$f_1 \neq f_2$, then the edges involved in~$S_1$ share at most one endpoint with the edges involved in~$S_2$ (namely a common endpoint
    between two edges in~$T_1$), and the conclusion follows
    from \Cref{lemma:at most one endpoint}.
    Moreover, if~$e_1 = e_2$ and~$f_1 = f_2$, then~$S_1 = S_2$, so we can ignore this case.

    It remains to consider the case that~$S_1$ and~$S_2$ share exactly one
    tour-edge. The conclusion then follows from \Cref{lemma:share one chord disjoint}.
\end{proof}

\begin{lemma}\label{lemma:chord disjoint ordering}
    Let~$S_1 = S_T(e_1, f_1)$ and~$S_2 = S_T(e_2, f_2)$
    be two 2-changes with the property that
    three edges of~$\{e_1, e_2, f_1, f_2\}$ form a path disjoint
    from the remaining edge. Then~$S_1$ and~$S_2$ are chord-disjoint.
\end{lemma}

\begin{proof}
    Observe that two edges on the path~$P$ must form a 2-change~$S$
    together. The remaining edge of the path forms a 2-change~$S'$ with some edge~$e'$
    vertex-disjoint from the path. It follows that these 2-changes are chord-disjoint,
    since the chord-edges of~$S'$ both contain vertices not on~$P$
    while the chord-edges of~$S$ only contain vertices of~$P$.
\end{proof}

\begin{figure}
    \centering
    \begin{tikzpicture}[scale=1.2]

    \def\radius{2}
    
    \def\numPoints{8}
    \def\numPointstwo{17}
    
    \def\numDots{17}
    \def\ratio{(\numPointstwo-1)/\numPointstwo}

    \foreach \i in {1,...,\numPointstwo} {
            \draw[gray!80] ({360/\numPointstwo * (\i - 1)}:\radius) arc ({360/\numPointstwo * (\i - 1)}:{360/\numPointstwo * \i}:\radius);
    }
    
    \foreach \i in {1,...,\numPoints} {
        \draw[gray, dotted] (0,0) -- ({\ratio*360/\numPoints * (\i - 1)}:\radius);
         \node[] at (({\ratio*360/\numPoints * (\i + 0.5 - 1)}:1.15*\radius) {$T_{\i}$};
        \ifodd\i
            \draw[very thick, black] ({\ratio*360/\numPoints * (\i - 1)}:\radius) arc ({\ratio*360/\numPoints * (\i - 1)}:{\ratio*360/\numPoints * \i}:\radius);
            \foreach \j in {1,5,...,\numDots} {
        }
        \fi
    }
    \def\i{\numPointstwo}
    \draw[white, dashed]
        ({360/\numPointstwo * (\i - 1)}:\radius) arc ({360/\numPointstwo * (\i - 1)}:{360/\numPointstwo * \i}:\radius);
    
    \foreach \i in {1,...,\numDots} {
        \vertex[] (v\i) at ({360/\numDots * (\i - 1)}:\radius) {};
    }
            
    \end{tikzpicture}
    \caption{
    Colors of the edges in the tour~$T$ at stage~$t = 3$, for~$n = 2^4 + 1$. The dotted lines
    are drawn to show the boundaries of each segment~$T_i$ more clearly. The segments of the tour
    are numbered starting at the right and proceeding counterclockwise. The 2-changes we
    consider in the proof of \Cref{lemma:set chord disjoint} are then the 2-changes formed
    from the red edges in~$T_i$ (drawn black) and the blue edges in~$T_{i+1}$ (drawn gray) that appear
    in even positions along~$T$, for~$i$ odd. 
    Note that the last edge along~$T$ is drawn dashed 
    to indicate that it is not used in the construction of~$\S$.}
    \label{fig:stages}
\end{figure}
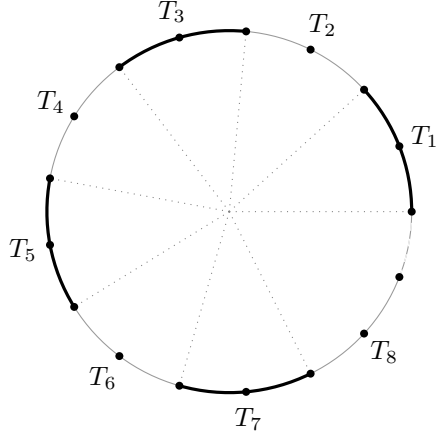

Next, we construct a set~$\mathcal{S}$ of 2-changes such that any
two 2-changes in~$\mathcal{S}$ are chord-disjoint, and most of
the edges of~$T$ participate in many 2-changes in~$\mathcal{S}$.
The construction we provide here works for complete
graphs with~$n = 2^k + 1$ vertices for some integer~$k$.

The construction of~$\mathcal{S}$ proceeds as follows.
Recall that we ordered the edges of~$T$ as~$(e_1, e_2, \ldots, e_n)$.
We define the following process on~$T$, occurring in stages.
At stage~$t$ we divide the tour into~$2^t$ equal segments~$\{T_1, \ldots, T_{2^t}\}$,
starting at~$e_1$.
In each stage, we color all the edges in each odd segment red and
the edges in each even segment blue.
The only exception to this rule is the last edge~$e_n$, which we color black;
this edge is not used to form any 2-changes.
See \Cref{fig:stages} for an illustration at stage~$t = 3$.

At each stage we consider the 2-changes formed by the red edges
in each odd segment~$T_i$ together with the \emph{even} 
blue edges in its successor segment~$T_{i+1}$.
We say that these are the 2-changes \emph{added} in stage~$t$, and denote the set of
these 2-changes by~$\mathcal{S}_t$. We continue this process for~$\log (n-1) - 1$ stages.
Note that in the final stage, each segment contains two colored edges.

\begin{lemma}\label{lemma:set chord disjoint}
    The 2-changes in~$\mathcal{S} =\bigcup_{t=1}^{\log n - 1}\mathcal{S}_t$ are chord-disjoint.
\end{lemma}

\begin{proof}
    For any two 2-changes in~$\mathcal{S}_t$, chord-disjointness
    follows from \Cref{lemma:successive chord disjointness}.
    It thus remains to show this for some~$S_1 \in \mathcal{S}_{t_1}$ and~$S_2 \in \mathcal{S}_{t_2}$. Assume
    w.l.o.g.\ that~$t_1 < t_2$. In the following, we write~$S_1 = S_T(e_1, f_1)$
    and~$S_2 = S_T(e_2, f_2)$.
    
    We consider the number of shared vertices between~$S_1$ and~$S_2$.
    Since the case of two shared vertices is the hardest to analyze, we consider it last.
    
    By \Cref{lemma:at most one endpoint}, if~$S_1$ and~$S_2$
    share at most one vertex, then the 2-changes are chord-disjoint.
    If~$S_1$ and~$S_2$ share three vertices, then they have
    exactly one edge in common. By \Cref{lemma:share one chord disjoint},
    they are then chord-disjoint.
    
    Next, if the two 2-changes have four vertices in common, then
    either the edges of~$S_1$ and of~$S_2$ form a cycle, or~$S_1 = S_2$.
    The former case is only possible if~$n = 4$, but we assume that~$n$ is odd throughout.
    Thus, assume~$S_1 = S_2$.
    By construction,~$S_1$ and~$S_2$ both remove one red and
    one even blue edge from successive segments. Observe that at any stage~$t \geq 2$,
    if~$e$ and~$f$ participate in the same 2-change,
    then by construction~$e$ and~$f$ are both red in every stage~$t' \leq t$.
    Thus both edges removed in~$S_2$ are
    red in stage~$t_1$. This is impossible, as~$S_1$ removed the same pair
    of edges as~$S_2$ in stage~$t_1$, 
    and every 2-change of~$\S_{t_1}$ removes one red and one blue edge.

    If~$S_1$ and~$S_2$ have three vertices in common, then the edges must form
    a path. This path must be~$e_1 e_2 f_1 f_2$ up to interchanging the indices, since
    the tour-edges of a 2-change cannot share any vertices. By our chosen orientation~$e_1$ must be red in stage~$t_1$ and~$e_2$ must be red in stage~$t_2$.
    Then~$f_1$ and~$f_2$ are both blue in these respective stages. But blue edges
    are only used in a 2-change when they are even, and~$f_1$ and~$f_2$ cannot
    both be even: a contradiction.

    Lastly, assume that~$S_1$ and~$S_2$ share two vertices. If these vertices
    are both
    incident to the
    same edge in both 2-changes, then~$S_1$ and~$S_2$ share exactly one edge,
    and so once again by \Cref{lemma:share one chord disjoint}~$S_1$ and~$S_2$ are chord-disjoint.
    Thus, there are two cases: either~$e_1$ and~$e_2$ share an endpoint,
    as do~$f_1$ and~$f_2$; or three edges among~$\{e_1, f_1, e_2, f_2\}$
    lie on a path disjoint from the last edge. The second case
    follows directly from \Cref{lemma:chord disjoint ordering}.

    In the first case, assume w.l.o.g.\ that~$e_1$ comes before~$e_2$ in~$T$
    and that~$e_1$ is red in stage~$t_1$. Then~$f_1$ is blue in stage~$t_1$,
    and thus~$f_1$ is an even edge. Since~$e_2$ directly follows~$e_1$ and
    we assume~$f_1$ and~$f_2$ share a vertex, we also know that~$e_2$ is red in stage~$t_2$ and~$f_2$ is blue in stage~$t_2$.
    But in constructing~$\S$ blue edges are only used in 2-changes
    when they occur in even positions along~$T$, and~$f_1$ and~$f_2$ cannot both be even;
    a contradiction.
    Note that this part of the analysis does not use the fact that~$t_1 < t_2$, and thus if we
    interchange the indices, the same reasoning goes through. Hence, this case
    can be excluded.

    This concludes the case analysis.
    We have thus shown that two 2-changes from two different stages must be
    chord-disjoint, and therefore~$\mathcal{S}$ consists only of chord-disjoint 2-changes.
\end{proof}

We now determine how often each edge of~$T$ is used in~$\S$.

\begin{lemma}\label{lemma:counting edge occurrences}
    For~$\S$ as constructed above, we have~$\{k_e(\S)\}_{e \in T} = {\{0, 1, 2,\ldots, n-3\}}$.
    Moreover, there are two edges with~$k_e(\S) = 0$ and two edges
    with~$k_e(\S) = \frac{n-1}{2}$.
\end{lemma}

\begin{proof}
    We maintain the same order on the edges of~$T$
    as used in the construction of~$\S$. Using this order, we call an edge
    even if it appears in an even position in this order, and odd otherwise.
    Note that the last edge~$e_n$ is not considered in the construction of~$\S$,
    and so~$k_{e_n} = 0$.
    For the remainder of the proof, we consider~$T$ with~$e_n$ removed.

    For convenience, denote by~$k_e^t$ the number of 2-changes that~$e$ participates in at stage~$t$,
    so that~$k_e = \sum_{t=1}^{\log n - 1} k_e^t$.
    We observe the following property of~$\S$.
    
    \begin{fact}
        Consider stage~$t$ in the construction of~$\mathcal{S}$. If an edge~$e$ is red in stage~$t$,
        then~$k_e^t = n/2^{t+1}$. If~$e$ is blue, then~$k_e^t = n/2^t$ if~$e$ is even, and~$k_e^t = 0$ otherwise.
    \end{fact}

    For any edge of~$T$, we can count the number of 2-changes of~$\mathcal{S}$
    it participates in as follows.
    Construct a rooted directed binary tree~$H$, where the nodes of~$H$
    are sub-paths of~$T$. The root of~$H$ is~$T$ itself, while the children of
    any node~$P$ of~$H$ are the first and the second halves of the edges in~$P$
    under the ordering of the edges of~$T$, labelled~$P_1$ and~$P_2$ respectively.
    Thus, the children of~$T$ are~$P_1 = \{e_1, \ldots, e_{{(n-1)}/2}\}$ and~$P_2 = \{e_{(n-1)/2+1}, \ldots, e_{n-1}\}$.
    We moreover label the arcs of~$H$. If~$P_1$ and~$P_2$ are the children of a node~$P$
    as described above, then the arc~$a_1 = (P, P_1)$ gets a label~$L(a_1) = 1$,
    while~$a_2 = (P, P_2)$ gets a label~$L(a_2) = 0$.

    We construct~$H$ in this way from the root, continuing until~$H$ has 
    depth~$\log (n-1) - 1$.
    (We define the depth of a node~$v$ in a rooted directed tree as the number of arcs 
    in the directed path from the root to~$v$. The depth of the tree itself is the 
    largest depth among all nodes of the tree.)
    Then the nodes at depth~$t$ of~$H$ are the parts into which~$T$ is partitioned
    at stage~$t$. Note that the leaves of~$H$ each contain two successive edges of~$T$.

    Let~$a = (P, Q)$ be an arc in~$H$ with~$L(a) = 1$.
    From the construction of~$\S$, it follows
    that any edge~$e \in Q$ is colored red in the stage corresponding to the depth of~$Q$.
    Conversely, if~$L(a) = 0$, then~$Q$ is colored
    blue in this stage.

    For each leaf~$P$ of~$H$, we then consider the path~$P(v) = T a_1 P_1 a_2\ldots a_{\log (n-1) - 1} Q$
    from the root to this leaf.
    Following this path, we collect the labels of the arcs along this path into a
    string~$x(Q)$, so~$x(Q) = L(a_1) L(a_2) \ldots L(a_{\log (n-1) - 1})$.
    For~$e \in Q$, we set~$x_e = x(Q)$. 

    Given~$x_e$ for some edge~$e \in T$, we know that~$e$ is red in stage~$t$
    if and only if~$x_e(t) = 1$. We thus know exactly in which stages~$e$ is colored red, and in which it is colored blue.
    Using this information together with the fact above,
    we can derive formulae for~$k_e$ for any~$e \in T$. 
    There are two distinct cases.
    
    \begin{description}
        \item {\bfseries Case 1:~$e$ odd.} 
            Since~$e$ is odd, it only participates in any 2-change when it is colored red. 
            Thus, stage~$t$ contributes~$x_e(t) \cdot (n-1)/2^{t+1}$ 2-changes, and so we count
            \begin{align*}
                k_e &= (n-1)\sum_{t=1}^{\log (n-1) - 1} x_e(t) \cdot \frac{1}{2^{t+1}}
                 = \sum_{t=1}^{\log (n-1) - 1} x_e(t) \cdot 2^{\log (n-1) - 1 -t}\\
                 &= \sum_{j=0}^{\log (n-1) - 2} x_e(\log (n-1) - 1 - j)\cdot 2^{j}.
            \end{align*}
            Note that for a given bit string~$x_e$, this is simply the decimal expansion of
            the binary number represented by~$x_e$.
            
            Since every bit string of length~$\log (n-1) - 1$ is present for some
            odd edge (there are~$2^{\log (n-1) -1}$ leaves in~$H$ and each leaf corresponds to a 
            distinct string), we find that~$\{k_e \mid \text{odd } e \in T\}
            = \{0, 1, \ldots, (n-1)/2-1\}$. 
        \item {\bfseries Case 2:~$e$ even.} An even edge contributes~$(n-1)/2^{t+1}$ 2-changes at stage~$t$
            if it is red, and~$(n-1)/2^t$ 2-changes if it is blue. Thus,
            the contribution at this stage is~$(n-1)/2^{t} - x_e(t) \cdot (n-1) / 2^{t+1}$, and so
            we have
            \begin{align*}
                k_e &= (n-1) \left(\sum_{t=1}^{\log (n-1) - 1} 2^{-t} - \sum_{t=2}^{\log (n-1) - 1}
                    x_e(t)\cdot \frac{1}{2^{t+1}}\right) \\
                    &= n - 3 - 
                       \frac{n-1}{2} \cdot \sum_{t=1}^{\log (n-1) - 1} x_e(t)\cdot \frac{1}{2^{t}} \\
                    &= n - 3 - \sum_{j=0}^{\log (n-1) - 2} x_e(\log (n-1) - 1 - j)\cdot 2^j.
            \end{align*}
            As in the previous case, we recognize here the decimal expansion
            of~$x_e$ in the second term. Thus,~$\{k_e \mid \text{even } e \in T\} =
            \{(n-1)/2-1, (n-1)/2, \ldots, n-2, n-3\}$.
            \qedhere
    \end{description}
\end{proof}

\subsection{Putting the Pieces Together}

We return to the positive orthant probability
by constructing an inverse covariance matrix~$\Sigma^{-1}$ corresponding to~$\S$
according to \Cref{lemma: rewrite to orthant probability}.
In the following, we sort the edges of~$T$ by decreasing value of~$k_e(\S)$.

Observe that the first~$d = (n-3)/2$
edges of~$T$ in this order each form 2-changes with one another in~$\S$.
Note that~$d$ is an integer, as~$n = 2^k + 1$ is odd.
The entries of~$\Sigma^{-1}$ corresponding to these 2-changes are~$\Sigma_{ij}^{-1} = 1/\sqrt{k_i k_j} \geq \frac{1}{n-3}$.
Thus, the inverse covariance
matrix~$\Sigma^{-1}$ constructed from~$\S$ is upper-left triangular,
except with ones on the diagonal. We can then write it in block form,
\[
    \Sigma^{-1} = \begin{pmatrix}
        \tilde{\Sigma}^{-1} & A \\
        A^T & I
    \end{pmatrix},
\]
where the diagonal entries of~$\tilde{\Sigma}^{-1} \in \mathbb{R}^{d \times d}$
are each 1, and the off-diagonal entries
each satisfy~$\tilde{\Sigma}^{-1}_{ij} = \geq \frac{1}{n-3} = \frac{1}{2d}$.
We consider the matrix~$\hat{\Sigma}^{-1}$, which is identical to~$\tilde{\Sigma}^{-1}$, but
with the off-diagonal entries replaced by~$\frac{1}{2d}$.

Comparing~$\hat{\Sigma}^{-1}$ to \Cref{lemma: rewrite to orthant probability},
we proceed to bound the positive orthant probability associated with~$\hat{\Sigma}^{-1}$.
A trivial bound for this probability is~$2^{-d}$.
\Cref{lemma: reduced orthant probability} therefore represents a non-trivial,
if modest, improvement.

\begin{lemma}\label{lemma: reduced orthant probability}
    Let~$X$ be distributed according to~$\mathcal{N}_d(0, \hat{\Sigma})$,
    where~$\hat \Sigma^{-1} \in \real^{d \times d}$ has unit diagonal entries
    and off-diagonal entries~$\frac{1}{2d}$.
    The positive orthant probability of~$X$ is
    bounded from above by~$O\left(2^{-d}e^{-\frac{2d}{9\pi}}\right)$.
\end{lemma}

\begin{proof}
    It suffices to bound~$\expect[\|X\|^2 \given \real_+^d]$ and use
    \Cref{thm: orthant expectation}.
    By symmetry of~$\hat \Sigma^{-1}$ the marginal densities of the entries of~$X$
    are all identical. Hence,
    \[
        \expect[\|X\|^2 \given \real_+^d] = d \cdot \expect[X_1^2 \given \real_+^d].
    \]
    To compute this expected value we 
    use \Cref{thm: manjunath and wilhelm}, which yields
    \[
        \expect[X_i^2 \given \real_+^d] = \sigma_{ii} + \sum_{k=1}^d \sum_{q \neq k} \sigma_{ik}
            \left(
                \sigma_{iq} - \frac{\sigma_{kq}\sigma_{ik}}{\sigma_{kk}}
            \right)F_{kq}(0, 0),
    \]
    where~$\sigma_{ij}$ denotes the~$i,j$-entry of~$\hat{\Sigma}$, and~$F_{kq}(x, y)$ is the joint marginal
    distribution of~$(X_k, X_q)$ conditional on~$X \in \real_+^d$. 

    The marginal distributions are given by a rather compact expression, since we
    only evaluate them at the origin. Using the conditional density of~$X$ as given in
    \Cref{eq: truncated normal},
    \[
        F_{kq}(0, 0) = \frac{1}{(2\pi)^{d/2}\sqrt{\det{\hat \Sigma}}} \cdot \frac{1}{\prob(\real_+^d)}\cdot
        \int_{\real_+^{d-2}} 
            \exp\left(
                -\frac{1}{2}x^T \hat\Sigma_{(kq)}^{-1} x
            \right) \dd x,
    \]
    where~$\hat\Sigma^{-1}_{(kq)}$ is simply~$\hat \Sigma^{-1}$ with the~$k^\text{th}$ and~$q^\text{th}$ rows
    and columns removed.
    
    This expression can be simplified further as follows.
    For arbitrary~$m \in \mathbb{N}$, 
    let~$\Sigma^{-1}_m$ be the~$m \times m$ matrix with unit diagonal and
    off-diagonal entries~$\frac{1}{2d}$.
    Then~$\hat{\Sigma}^{-1} = \Sigma^{-1}_d$ and~$\Sigma_{(kq)}^{-1} = \Sigma_{d-2}^{-1}$.
    Hence, inserting the expression for~$\prob(\real_+^d)$ and cancelling like terms,
    we obtain
    \begin{align*}
        F_{kq}(0, 0) &= 
        \frac{\int_{\real_+^{d-2}} \exp\left(-\frac{1}{2}x^T \Sigma^{-1}_{d-2}x\right)\dd x}
            {\int_{\real_+^d} \exp\left(-\frac{1}{2}x^T \Sigma^{-1}_d x\right)}
                \geq 
        \frac{\int_{\real_+^{d-2}} \exp\left(-\frac{1}{2}x^T \Sigma^{-1}_{d-2}x\right)\dd x}
            {\int_{\real_+^{d-2}} \exp\left(-\frac{1}{2}x^T \Sigma^{-1}_{d-2} x\right)
                \left(\int_0^\infty e^{-\frac{1}{2}x^2} \dd x\right)^2}\\
             &= \frac{2}{\pi}.
    \end{align*}

    Next, we need to compute the entries of~$\hat \Sigma$ and the determinant~$\det \Sigma_{d}^{-1}$.
    Note that~$\Sigma_d^{-1} = D + \frac{1}{2d}ee^T$, where~$D = \left(1 - \frac{1}{2d}\right)I$
    and~$e$ is the all-1 column vector.
    It can then be straightforwardly verified from the Sherman-Morrison formula that
    \begin{align*}
        \det{\Sigma_d^{-1}} = \frac{3d - 1}{2d - 1} \left(\frac{2d - 1}{2d}\right)^d 
        \quad \text{and} \quad
        \sigma_{ij} = \frac{2d}{2d-1} \cdot \begin{cases}
                1 - \frac{1}{3d-1}, & \text{if~$i = j$}, \\
                -\frac{1}{3d - 1}, & \text{otherwise}.
        \end{cases}
    \end{align*}
    We omit the details of the calculations as they are routine. 

    Let~$g_{kq} = \sigma_{1k}\left(\sigma_{1q} - \frac{\sigma_{kq}\sigma_{1k}}{\sigma_{kk}}\right)$.
    We now compute~$\sum_{k=1}^d \sum_{k \neq q} g_{kq}$. There are three values that~$g_{kq}$ takes:
    \begin{description}
        \item[Case 1:~$k \neq 1$,~$q \neq 1$.]  Then~$g_{kq} = \left(\frac{2d}{2d-1}\right)^2
                \left(\frac{1}{3d-1}\right)^2\left(1 + \frac{1}{3d-2}\right)$.
        \item[Case 2:~$k = 1$.] Then~$g_{kq} = \sigma_{1q} - \sigma_{1q}\sigma_{11}/\sigma_{11} = 0$.
        \item[Case 3:~$q = 1$.] Then~$g_{kq} = -\left(\frac{2d}{2d-1}\right)^2 \cdot \frac{1}{3d-1}
            \cdot \left(1 - \frac{1}{3d-1} - \frac{1}{3d-2}\right)$. Simplifying slightly,
            we have~$g_{kq} \leq -\left(\frac{2d}{2d-1}\right)^2 \cdot \frac{1}{3d-1}
            \cdot \left(1 - \frac{2}{3d-2}\right)$.
    \end{description}
    In computing the sum over~$k$ and~$q$ we find~$(d-1)(d-2)$ terms corresponding to
    case Case 1,~$d-1$ terms corresponding to Case 2,
    and~$d-1$ terms corresponding to Case 3. Thus, we have
    \begin{align*}
        \sum_{k=1}^d \sum_{q\neq k} g_{kq} \leq\, & (d-1)(d-2) \cdot 
            \left(\frac{2d}{2d-1}\right)^2 \cdot 
                    \left(\frac{1}{3d-1}\right)^2\left(1 + \frac{1}{3d-2}\right) \\
           &- (d-1) \cdot \left(\frac{2d}{2d-1}\right)^2 \cdot \frac{1}{3d-1}
            \cdot \left(1 - \frac{2}{3d-2}\right).
    \end{align*}
    Rearranging, simplifying, and grouping negligible terms, we find
    \begin{align*}
        \sum_{k=1}^d \sum_{q\neq k} g_{kq} \leq - \frac{2}{9} + O\left(\frac{1}{d^2}\right).
    \end{align*}
    
    Plugging all of the above into the formula for~$\expect[X_1^2 \given \real_+^d]$ results in
    \begin{align*}
        \expect[X_1^2\given \real_+^d] \leq \sigma_{11} - \frac{2}{\pi}
            \left[
               \frac{2}{9} - O\left(\frac{1}{d^2}\right)
            \right] = 
            1 - \frac{4}{9\pi} + O\left(\frac{1}{d}\right).
    \end{align*}
    Using \Cref{thm: orthant expectation}, we then finally obtain
    \[
        \prob(\real_+^d) \leq 2^{-d+1} e^{-d/2} 
            e^{\frac{1}{2}\left(O(1) + d - \frac{4d}{9\pi}\right)}
            = O\left(2^{-d}e^{-\frac{2d}{9\pi}}\right)
    \]
    as claimed.
\end{proof}

The following is now an immediate consequence of
\Cref{lemma: rewrite to orthant probability,lemma: reduced orthant probability},
recalling that~$d = (n-3)/2$.

\begin{lemma}\label{lemma: Eg bound}
    For~$\S$ as constructed above,~$\G(\S) = O\left(e^{-\frac{n}{9\pi}}\right)$.
\end{lemma}

\begin{proof}
    By \Cref{lemma: rewrite to orthant probability}, we need to
    multiply the bound from \Cref{lemma: reduced orthant probability}
    by a factor~$2^{d}\sqrt{\det {\hat\Sigma}}$. It is easily seen from the proof
    of the latter lemma that this determinant is~$O(1)$.
\end{proof}

Finally, we combine all of the above for the last crucial lemma, from
which our main result follows directly.

\begin{lemma}\label{lemma:probability of 2-optimal}
    Let~$G$ be a complete graph on~$n =2^k+1 \geq 5$ vertices, with edge weights
    drawn independently from~$U[0,1]$ for each edge.
    Let~$T$ be any tour through~$G$.
    The probability that~$T$ is 2-optimal is bounded from above by 
    \[
        O \left(\frac{c^n}{\sqrt{(n-2)!}}\right),
    \]
    where~$c = \sqrt{\frac{\pi}{2}} \cdot e^{-\frac{1}{9\pi}} < 1.2098$.
\end{lemma}

\begin{proof}
    Let~$\mathcal{S}$ be a set of 2-changes on~$G$ constructed as described above, and
    let~$k_e$ be the number of 2-changes in~$\mathcal{S}$ in which~$e \in T$
    participates. 
    From \Cref{lemma:counting edge occurrences}
    we have
    \begin{align*}
            \prod_{\substack{e \in T \\ k_e > 0}} \frac{1}{k_e}
                = \frac{1}{\frac{n-1}{2} - 1} \cdot \frac{1}{(n-3)!} = O\left(\frac{1}{(n-2)!}\right).
    \end{align*}
    Using this result in
    \Cref{lemma:2-optimal chord-disjoint} together with
    \Cref{lemma: Eg bound} yields the claim.
\end{proof}

This leads to our last result (\Cref{thm:count_2opt}), using the fact that
the number of tours on a complete graph is~$\frac{1}{2}(n-1)!$.

\subsection{Numerical Experiment}
\label{sec: numerical}

It seems unlikely that the bound in~\Cref{thm:count_2opt} is tight. A simple numerical
estimate of~$\G(\S)$ already shows that it could be improved to~$O(1.0223^n\sqrt{n!})$,
but even this may be a coarse approximation; after all, in constructing~$\S$, we discard
many 2-changes.

To estimate the number of 2-optimal tours numerically, one could take a na\"ive approach:
simply fix a tour~$T$, generate edge weights from~$U[0,1]$, and check 2-optimality
of~$T$ with these weights. By repeating this experiment we can 
estimate~$\prob(\text{$T$ is 2-optimal})$.
However, since this probability is super-exponentially small (\Cref{lemma:probability of 2-optimal}),
this is rather inefficient.

We can do better by taking a different view of the problem. Fix again an arbitrary tour~$T$.
We write the edge weights as a vector~$w \in [0, 1]^E$. Given~$T$, we can determine
all~$n(n-3)/2$ possible 2-changes on~$T$. Local optimality of~$T$
means that all these 2-changes yield a negative improvement. Thus, for a 2-change
that removes edges~$e_1$ and~$e_2$ and adds~$f_1$ and~$f_2$, this yields an inequality
\begin{align}\label{eq:2-opt ineq}
    w_{e_1} + w_{e_2} - w_{f_1} - w_{f_2} \leq 0.
\end{align}
Each possible 2-change on~$T$ yields such a constraint. Together with the 
constraints~$0 \leq w_e \leq 1$ for each edge, this yields a convex polytope~$P$
embedded in~$\real^m$. Since the edge weights are i.i.d\ uniform
random variables,
\[
    \prob(\text{$T$ is 2-optimal}) = \mathrm{vol}(P),
\]
the~$|E|$-dimensional volume of~$P$.

We remark now that our approach through \Cref{lemma:2-optimal chord-disjoint} is equivalent
to bounding the volume of a polytope that contains~$P$, by discarding some of the
inequalities (\Cref{eq:2-opt ineq}). It may be possible to use methods from convex geometry to
compute better estimates of~$\mathrm{vol}(P)$; we leave this discussion 
to \Cref{sec:counting_discussion}.

Computing the volume of an arbitrary polytope is itself not an easy task: this
problem is \sharpP{}-hard in general, as shown by Dyer
and Frieze~\cite{dyerComplexityComputingVolume1988}. In the same work however, they show
that volume approximation can be done efficiently by randomized algorithms.
We use the open-source software package \emph{Volesti}~\cite{chalkisVolestiVolumeApproximation2021,emirisPracticalPolytopeVolume2018}
to numerically approximate the volume of~$P$. 

For a range of different~$n$,
we compute the inequalities that define~$P$, and output them in a format readable
for Volesti. We use the `Cooling Balls' algorithm of Volesti, with an error parameter
of 0.001 (see the specification of Volesti \cite{chalkisVolestiVolumeApproximation2021}
approximated~$\mathrm{vol}(P)$ up to~$n = 40$. The results of these
computations are shown in \Cref{fig:volume}.

\begin{figure}
    \centering
    \includegraphics[width=0.75\linewidth]{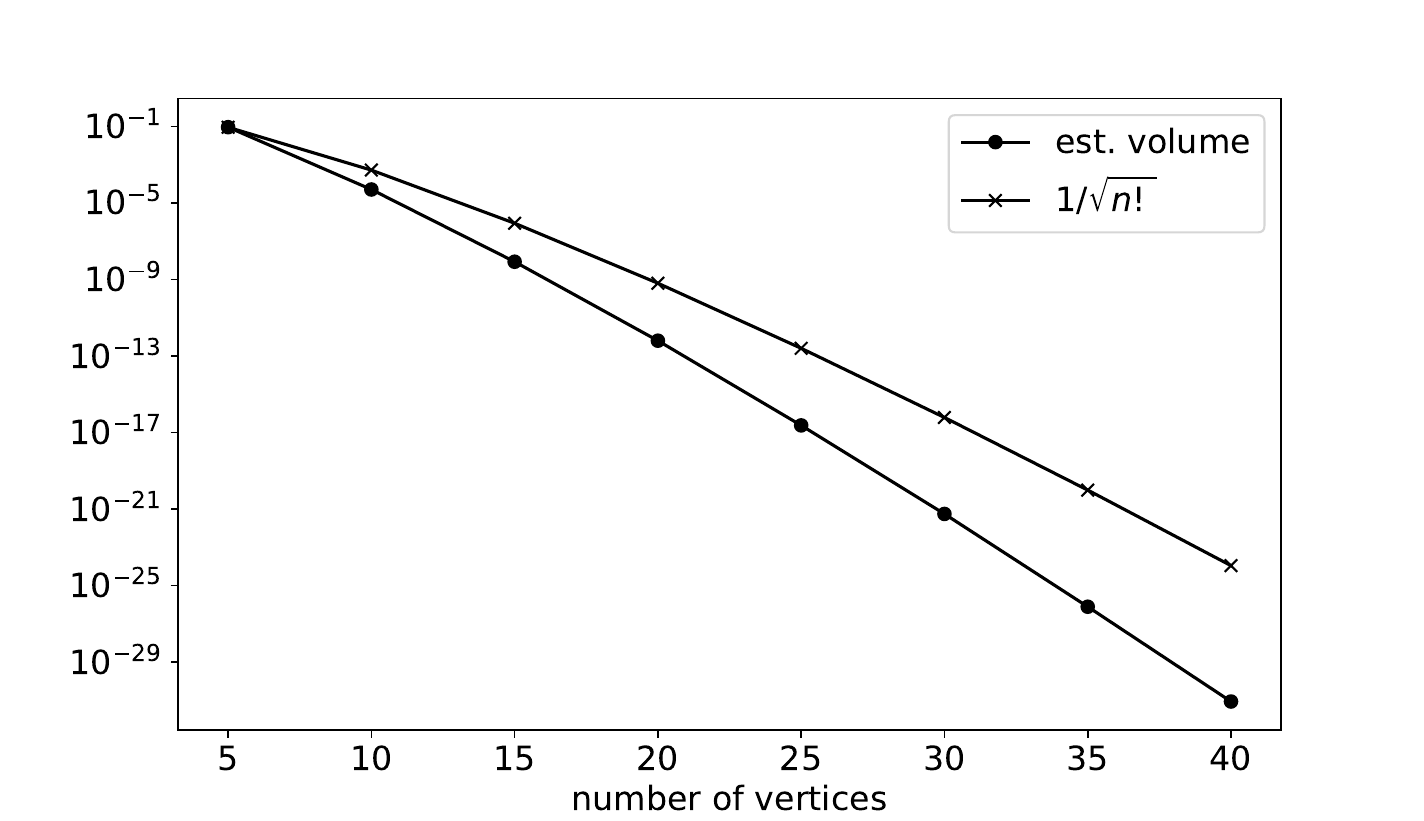}
    \caption{Estimated volume of the 2-opt polytope $P_n$, for different values of $n$, as
    computed by Volesti. For comparison, the function $n \mapsto 1/\sqrt{n!}$ is also
    plotted.}
    \label{fig:volume}
\end{figure}

\section{Discussion}
\label{sec:counting_discussion}

The result we present in \Cref{thm: sharp P complete} is to our knowledge
the first hardness result for counting locally optimal solutions for
a \emph{natural} local search problem. We note that it is easy to show that
counting the local optima is \sharpP{}-hard for \emph{some} local search problem. For instance,
Johnson et al.~\cite{johnsonHowEasyLocal1988} provided a local search problem whose
locally optimal solutions are exactly the satisfying assignments of an instance of
Boolean Satisfiability (plus one extra solution). However, their construction is rather
artificial, whereas 2-opt is a heuristic often used in practice.

\Cref{thm:count_2opt} is to our knowledge the first non-trivial bound on
the number of locally optimal solutions for a local search problem.
It must be noted that we only proved \Cref{thm:count_2opt} for specific values
of~$n$, namely~$n = 2^k + 1$ for~$k \in \mathbb{N}$. This restriction simplifies
mainly the construction of~$\S$, our set of chord-disjoint 2-changes. We believe
that the results can be extended at the cost of some complexity in the proofs.
For the sake of simplicity, we chose not to do so.

The bound in \Cref{thm:count_2opt} shows that in expectation, the number of 2-optimal
tours in a random TSP instance is approximately the square root of the total number of tours.
While still a rather large number, it is nonetheless a super-exponentially small
fraction of the total number of tours. 

\paragraph{Limitations.} 
A natural question concerns the tightness of the bound in \Cref{thm:count_2opt}.
We have not succeeded in constructing a lower bound for the number
of 2-optimal tours. Presumably such a construction would be rather involved,
since there is little structure in this random instance model.
Nonetheless, we can still argue that \Cref{thm:count_2opt} can be improved
significantly from two directions: bounding~$\G(\S)$ and
constructing~$\S$.

\Cref{lemma: Eg bound} is far from optimal. A simple numerical experiment to estimate~$\G(\S)$ yields~$\G(\S) = O(1.226^{-n})$. 
This results in a bound of~$O(1.0223^n \sqrt{n!})$ in \Cref{thm:count_2opt}. We also note that, even though our set~$\S$
achieves essentially the largest number of chord-disjoint 2-changes possible on any tour,
it may not be an optimal choice: different choices may lead to
better sequences of values for~$k_e$.

Using the trivial bound of~$\G(\S) \leq 1$
instead of \Cref{lemma: Eg bound} in \Cref{lemma:probability of 2-optimal}, we would
obtain in \Cref{thm:count_2opt} a bound of approximately~$O(1.2534^n \sqrt{n!})$.
The difference with our bound 
is thus less than a factor of~$1.04^n$, which is a rather minute improvement
for the amount of trouble we went through to obtain it. We therefore regard 
the calculations leading to \Cref{lemma: Eg bound} more as a proof-of-concept that
significant improvements are still possible to obtain
rigorously, and as a demonstration of the effort required
to achieve anything non-trivial.

Another possible direction for improving the bound in \Cref{thm:count_2opt} lies in
polyhedron volume estimation. In \Cref{sec: numerical}, we defined a polytope~$P$
such that the volume of~$P$ is the probability that a tour on~$n$ vertices is 2-optimal.
While computing the volume of a polytope is \sharpP{}-hard in general~\cite{dyerComplexityComputingVolume1988},
it may be possible to obtain an asymptotic formula for our specific case,
as was done by Canfield and McKay for
the well-studied Birkhoff polytope~\cite{canfieldAsymptoticVolumeBirkhoff2009}.

\paragraph{A conjecture.}
In light of these observations, and the estimate resulting from the numerical experiments,
we conjecture the following.

\begin{conjecture}
    Let~$G$ be a complete graph on~$n$ vertices, with edge weights
    drawn independently from~$U[0,1]$ for each edge. Then the expected number of
    2-optimal tours on~$G$ is bounded from above by~$O(\sqrt{n!})$.
\end{conjecture}

The numerical data suggests that the expected number of tours may even be~$c^n \sqrt{n!}$
for some~$c < 1$ (as opposed to~$c > 1$ as in \Cref{thm:count_2opt}), although we could not
perform the numerical experiments for large enough~$n$ to state this with confidence.

\bibliographystyle{abbrv}
\bibliography{bibliography.bib}

\end{document}